\documentclass[11pt]{article}
\usepackage[utf8]{inputenc}
\usepackage{amsmath,amsfonts,amssymb,amsthm}
\usepackage[dvips]{graphicx}
\usepackage{algorithm}
\usepackage{algpseudocode}
\usepackage[margin=1in]{geometry}
\usepackage[pdftex,pagebackref,colorlinks=true,urlcolor=blue,linkcolor=blue,citecolor=blue,pdfstartview=FitH]{hyperref}

\usepackage{xspace}
\usepackage[dvipsnames]{xcolor}

\title{\textbf{Finding a Bounded-Degree Expander \\ Inside a Dense One}\footnote{Partially supported by the ERC Advanced Grant 788893 AMDROMA "Algorithmic and Mechanism Design Research in Online Markets" and MIUR PRIN project ALGADIMAR "Algorithms, Games, and Digital Markets".}}
\author{
    Luca Becchetti\\
    {\scriptsize{}Sapienza Università di Roma}\\
    {\footnotesize{} Rome, Italy}\\
    {\footnotesize{}\texttt{becchetti@dis.uniroma1.it}} 
    \and Andrea Clementi \\  
    {\scriptsize{}Università di Roma Tor Vergata}\\
    {\footnotesize{} Rome, Italy}\\
    {\footnotesize{}\texttt{clementi@mat.uniroma2.it}} 
    \and Emanuele Natale\\
    {\scriptsize{}Universit\'e C\^ote d'Azur, CNRS, INRIA}\\
    {\footnotesize{}Sophia Antipolis, France}\\
    {\footnotesize{}\texttt{natale@unice.fr}}
    \and Francesco Pasquale \\
    {\scriptsize{}Università di Roma Tor Vergata}\\
    {\footnotesize{} Rome, Italy}\\
    {\footnotesize{}\texttt{pasquale@mat.uniroma2.it}} 
    \and Luca Trevisan \\
    {\scriptsize{}U.C. Berkeley}\\
    {\footnotesize{} Berkeley, CA, United States}\\
    {\footnotesize{}\texttt{luca@berkeley.edu}}
}
\date{}

\newtheorem{definition}{Definition}

\newtheorem{lemma}[definition]{Lemma}
\newtheorem{claim}[definition]{Claim}
\newtheorem{theorem}[definition]{Theorem}

\newcommand{\bigO}{\mathcal{O}}
\newcommand{\Prob}[2]{\mathbf{P}_{#1} \left( #2 \right)}

\newcommand{\Expec}[2]{\mathbf{E}_{#1} \left[ #2 \right]}

\newcommand{\polylog}[1]{\mbox{polylog}\left(#1\right)}
\newcommand{\poly}{ {\mathrm{poly}}}

\newcommand{\local}{{\sc{local}}}
\newcommand{\gossip}{{\sc{gossip}}}
\newcommand{\ALG}{{\sc raes}}
 \newcommand{\dout}{ {\mathrm{d_v^{out}}} }
\newcommand{\din}{ {\mathrm{d_v^{in}}}}
\newcommand{\grado}{\ensuremath{d}}

\newcommand{\bx}{\mathbf{x}}
\newcommand{\by}{\mathbf{y}}
\newcommand{\sM}{\mathcal{M}}
\newcommand{\SST}{\mathrm{SS}}

\newcommand{\sC}{\mathcal{C}}
\newcommand{\bin}{\mathbf{Cost}}
\newcommand{\dest}{\mathrm{Dest}}
\newcommand{\sRJ}{\mathrm{Rej}} 

\newcommand{\nreq}{\ensuremath{d}} 

\newcommand{\rc}{\ensuremath{rc}}
\newcommand{\rss}{\ensuremath{rss}}
\newcommand{\vol}{\ensuremath{vol}}
\newcommand{\out}{\ensuremath{\delta}}
\newcommand{\state}{\mathbf{X}}
\newcommand{\DEC}{\sc{Dec}}

\newcommand{\freq}{\ensuremath{\epsilon}} 

\renewcommand{\leq}{\leqslant}
\renewcommand{\le}{\leqslant}
\renewcommand{\geq}{\geqslant}
\renewcommand{\ge}{\geqslant}
\renewcommand{\epsilon}{\varepsilon}

\newcommand{\bY}{\mathbf{Y}}

\begin{document}


\maketitle
\begin{abstract}

It follows from the Marcus-Spielman-Srivastava proof of the Kadison-Singer
conjecture that if $G=(V,E)$ is a $\Delta$-regular dense expander then there is
an edge-induced subgraph $H=(V,E_H)$ of $G$ of constant maximum degree which is
also an expander. As with other consequences of the MSS theorem, it is not
clear how one would explicitly construct such a subgraph. 

We show that such a subgraph (although with quantitatively weaker expansion and
near-regularity properties than those predicted by MSS) can be constructed with
high probability in linear time, via a simple algorithm. Our algorithm allows a
distributed implementation that runs in $\bigO(\log n)$ rounds and does
$\bigO(n)$ total work with high probability. 

The analysis of the algorithm is complicated by the complex dependencies that
arise between edges and between choices made in different rounds. We sidestep
these difficulties by following the combinatorial approach of counting the
number of possible random choices of the algorithm which lead to failure. We do
so by a compression argument showing that such random choices can be encoded
with a non-trivial compression.

Our algorithm bears some similarity to the way agents construct a communication
graph in a peer-to-peer network, and, in the bipartite case, to the way agents
select servers in blockchain protocols. 

\end{abstract}

\medskip


\thispagestyle{empty}
 
\newpage
\setcounter{page}{1}

 
\section{Introduction}

\def\R{{\mathbb R}}

The proof of the Kadison-Singer conjecture by Marcus, Spielman and
Srivastava~\cite{MSS15.ks} (henceforth, the {\em MSS
Theorem}) has several important graph theoretic corollaries. In particular, if
$G=(V,E)$ is an undirected graph with $n$ nodes in which every edge has
effective resistance $\bigO(n/|E|)$, then there is an edge-induced subgraph
$H=(V,E_H)$ of $G$ that has $\bigO(n/\epsilon^2)$ edges and that is an
unweighted $\epsilon$-spectral-sparsifier\footnote{A weighted graph $H=(V,E_H)$
is an $\epsilon$-spectral-sparsifier \cite{BSS12} of a graph $G=(V,E_G)$ if,
for every vector $\bx\in \R^V$, we have \[ (1-\epsilon) \sum_{(u,v)\in E_G}
(x_u - x_v)^2 \leq \sum_{(u,v) \in E_H} w_H(u,v) \cdot (x_u - x_v)^2 \leq
(1+\epsilon) \sum_{(u,v)\in E_G} (x_u - x_v)^2 \] where $w_H(u,v)$ is the
weight of the edge $(u,v)$ in $H$. We say that $H$ is {\em unweighted} if the
weights of all the edges of $H$ are all equal to the same scaling factor $|E_G|
/ |E_H|$.} of $G$.
 
Interesting examples of graphs to which this statement applies are
edge-transitive graphs, such as the hypercube, and regular expanders of
constant normalized edge expansion. As with other consequences of the MSS
Theorem, and other non-constructive results proved with similar techniques, it
is not known how to construct such subgraphs in polynomial (or even
subexponential) time.

In the case of regular expanders, the result, qualitatively, states that if
$G=(V,E)$ is a $\Delta$-regular graph of constant normalized edge expansion,
there exists an edge-induced subgraph $H$ of $G$ that has constant maximum
degree and constant normalized edge expansion. 

In this work, we show how to constructively find such an $H$, assuming that
$\Delta = \Omega(n)$ and that the second eigenvalue of the adjacency matrix of
$G$ (which measures the spectral expansion of the graph) is at most a
sufficiently small constant times the degree $\Delta$.
The randomized algorithm we propose receives as input a $\Delta$-regular graph
$G$ and two integer parameters $d$ and $c$.

If we only assume $\Delta= \Theta(n)$,  $c > 2n/|E|$ and $d$ is a
sufficiently large absolute constant then, with high probability, 
the algorithm completes in $\bigO(n)$ steps and returns a subgraph $H$ of $G$, 
in which each node has degree between $d$ and $(c+1) \cdot d$ (see Theorem~\ref{thm:main-convergence}). 

If we further assume that the second eigenvalue of the adjacency matrix of $G$
is at most $\gamma\Delta$, with $\gamma$ a sufficiently small constant, we can
prove that, with high probability, $H$ has conductance $\Omega(1)$ (see
Theorem~\ref{thm:main-expanders}).

Our algorithm is extremely simple and naturally lends itself to a distributed implementation, in a model in which the underlying communication network is $G$ itself, with its nodes as computing elements. In this model, the nodes of $G$ can collectively identify a subgraph $H$ with the properties mentioned above in $\bigO(\log n)$  rounds and with $\bigO(n)$ total work and communication cost, in the sense that at the end of the protocol, each node knows its neighbors in $H$. 

The distributed version of our algorithm, that we call \ALG{} (for \emph{Request a link, then Accept if Enough Space}),
works in rounds, each consisting of two phases. Initially, each node has $0$ outgoing links and $0$ incoming links. In the first phase of each round, each node $v$ selects enough random neighbors (according to the topology of $G$) so that linking to all of them would secure $v$ a total of $d$ outgoing links. It then submits a request to each selected neighbor to establish a link. In the second phase of the round, each node accepts all requests received in the first phase of the current round, unless doing so would cause it to exceed the limit of $cd$ incoming links; if this is the case, the node rejects all requests it received in the first phase of the current round.  The algorithm completes when each node has established exactly $d$ outgoing links, so that no further requests are submitted. A formal description of the algorithm  is given in Section~\ref{se:prelim}. 

To show that our algorithm completes in $\bigO(\log n)$ rounds with high
probability when $G$ is $\Delta$-regular and $c > 2n/\Delta$, we show that, for any request submitted by some node $v$ in any round $t$, regardless of the remaining randomness of the algorithm, the request is accepted with probability at least $1/2$. This happens since, in each round, the number of nodes that reject any request is at most $n/2$. This is enough to show that convergence takes $\bigO(\log n)$ rounds with high probability and total work $\bigO(dn)$ on average. 
To prove that the total work is $\bigO(dn)$ with high probability we show that,
in each round $t$, if $\dout$ denotes the current number of $v$'s outgoing
links, $d\cdot n - \Expec{}{\sum_v \dout}$, i.e., the expected number of ``missing links'', shrinks, on average, by a constant factor. Moreover, the amount by which the above quantity changes at each step is a Lipschitz function of independent random variables, which means that we can argue with high
probability about the amount by which this quantity decreases.

The main result of this work is the proof that, if $G$ is a sufficiently good
expander, then the graph produced by the algorithm has constant expansion. 
In the spirit of how one analyzes the expansion of random regular graphs, we
would like to argue that, for every set $S\subseteq V$ of $s\leq n/2$ vertices,
there is at least a probability, say, $1 -  n^{-2} \cdot {n \choose s}^{-1}$,
that, of the $ds$ outgoing links from the vertices of $S$, at least
$\Omega(ds)$ are links from $S$ to $V-S$.  Then we could use a union bound over
all possible sets $S$ to say that with probability at least $1-1/n$ every set
$S$ has at least $\Omega(ds)$ links crossing the cut and going into $V-S$. The
probability distribution of the links created by the algorithm, and the ways in
which they are correlated, are however very difficult to analyze.

Our approach is to use a {\em compression argument}: we show that the random
choices of the algorithm that lead to a non-expanding graph can be
non-trivially compressed, and hence have low probability. The approach of
proving that an event is unlikely by showing that the random choices leading to
it are compressible is often a convenient way to analyze the outcome of an
algorithm. Such arguments are sometimes expressed in the language of {\em
Kolmogorov complexity}~\cite{VL97} and they are often used in cryptography to
analyze the security of protocols that involve a random oracle,
following~\cite{GT00}. In~\cite{MMR17}, the authors review various
probabilistic analyses that can be performed using compression argument (which
they call {\em encoding} arguments).

Our argument is roughly as follows: suppose that, in the graph constructed by
the algorithm, $S$ is a non-expanding set of vertices. If $G$ is a sufficiently
good $\Delta$-regular  expander, then, from the expander mixing lemma, we get
that the typical vertex of $S$ has only about $\Delta \cdot |S|/n$ neighbors in
$S$, but, if $S$ is non-expanding in $H$, then the typical node in $S$ has,
say, at least $.9 \cdot d$ of its $d$ outgoing links  in $S$. This means that,
for the typical node in $S$, we can represent $.9 d$ of its $d$ outgoing links
using $\log \frac {\Delta \cdot |S|}n$ bits instead of $\log \Delta$, with a
saving of order of $d |S| \log \frac n {|S|}$ bits. For sufficiently large
constant $d$, this is more than the $\log {n \choose |S|}$ bits that it takes
to represent the set $S$.  Unfortunately, things are not so easy because we
need the representations of choices made by the nodes in the algorithm to be
prefix-free, in order for their concatenation to be decodable. 
Therefore, we have to spend some additional bits in the representation 
of various terms, in particular for the choices that lead to
links from $S$ to $V-S$ (which are not so many since $S$ is a non-expanding
set) and for requests that are rejected. 
To complete the argument, we have to argue that the overall number of requests from nodes in $S$ that are rejected cannot be too large, for we would otherwise have a non-trivial way of compressing their description. This is true because, as argued above, each request has a small probability of being rejected, so that realizations of the random algorithm that lead to many rejected requests are unlikely,  hence compressible (for further details see Section \ref{ssec:overview}).

Algorithm \ALG\ is inspired by the way nodes create  bounded-degree overlay
networks in real-life distributed systems, such as peer-to-peer
protocols~\cite{GMS06,MSW06}  like BitTorrent, or in distributed ledger
protocols such as Bitcoin~\cite{N08}. In this protocol for example, each node
in a communication network is aware of the existence of a certain subset of the
other nodes (in our algorithm, for the generic node $v$ this subset corresponds to the set of $v$'s neighbors in $G$). Each node tries to establish a minimum number of connections to other nodes (or to special ``server'' nodes\footnote{In this setting, we notice that if $G$ is bipartite then $H$ is bipartite as well.}) and does so by selecting them at random from its list of known nodes (or known servers).  Nodes also have a maximum number of connections they are going to accept, rejecting further connections once this limit is reached.

On the other hand, our algorithm does not capture important traits of
peer-to-peer and blockchain models, such as the fact that nodes can join or
leave the network, and that nodes can exchange their lists of known nodes, so
that the graph ``$G$'' in fact is dynamic. We believe, however, that our
analysis addresses important aspects, such as the complicated dependencies that
arise between different links in the virtual network, and the {\em expansion}
properties of the resulting virtual network. Expansion in particular is closely
related to resilience to nodes leaving the network, a very important property
in practice.  

\section*{Related work}
\paragraph{Distributed constructions of expanders.} Our main result is an efficient, distributed algorithm to construct a bounded-degree expander. This question has been addressed for a number of models and initial conditions.
In~\cite{LS03}, Law and Siu  provide a distributed protocol running on the
\local\ asynchronous model  that  form expander graphs of arbitrary fixed
degree $d$. Their goal is to maintain the expansion property under insertions,
starting from a constant-size graphs, and they show how to do so in constant
time and constant message complexity per node insertion. See also~\cite{GMS06}
and~\cite{NW07} for such sequential constructions of expanders. 

In~\cite{ABLMO16}, Allen-Zhu et al. show a simple and local protocol that,
starting from any connected $d$-regular connected graph with $d = \Omega(\log
n)$, returns a $d$-regular expander. At every every round, an edge $e$ is
selected u.i.r. together with one  length-3 path including $e$ and, then, a
suitable flipping of the edges of this path is performed (so, the obtained
graph is not guaranteed to be a  subgraph of the original graph). Their
spectral analysis of the evolving graph shows that, after $\bigO(n^2d^2\polylog
n))$ rounds, the obtained random graphs is an expander, with high probability.
Their algorithm models the way in which nodes exchange neighborhood information
in real-life protocols, and it works starting from much more limited
information than ours (their initial information is an arbitrary graph of
logarithmic degree, while we start from a graph of linear degree which is
already an expander), and the price they pay is a polynomial, rather than
logarithmic, convergence time.
 
\paragraph{Sparsification.} We motivated our main result as a constructive
proof of a special case of the sparsification results implied by the MSS
theorem, for which no constructive proofs are known. Here it matters that we
are interested in sparsifying a regular graph by using an unweighted subgraph
of bounded maximum degree. If we allowed weighted graphs, and we were only
concerned about the average degree of the sparsifier, then an explicit
construction of constant average-degree sparsifiers for all graphs is given by
the BSS sparsifiers of~\cite{BSS12}. A parallel construction of the BSS
sparsifier, however, is not known. Parallel construction of (weighted,
unbounded max degree) sparsifiers have been studied~\cite{K16}, but such
constructions involve graphs of logarithmic average degree, a setting in which
our problem is trivial: given a $\Delta$-regular graph $G=(V,E)$, if we choose
each edge independently  with probability order of $(\log n)/\Delta$, we get a
graph that with high probability has maximum degree $\bigO(\log n)$ and, using
matrix Chernoff bound, we can show that it is a spectral sparsifier of $G$ if
$G$ is such that every edge has effective resistance $\bigO(n/|E|)$, including
the case of expanders and of edge-transitive graphs.
Finally, in   \cite{FM99},    Frieze and Molloy consider the task of  partitioning  expander graphs. In more detail, they provide a partitioning algorithm that, given as input a $\Delta$-regular graph with edge-expansion $\Phi$ and a parameter $k$, returns a partition $(E_1,...,E_k)$ of the edges,
such that each induced graph $G_i=(V,E_i)$ is almost-regular with node degree $\Theta(\frac \Delta k)$ and it has  edge expansion $\Omega(\Phi/k)$. 
Their algorithm runs in $\bigO(n^{\log \Delta})$ time and the required assumptions on $\Delta,k$ and $\phi$ do not allow to produce {\em constant}-degree subgraphs (their construction in fact requires $k = \bigO( \Delta /{\log \Delta})$).

\paragraph{Parallel Balls-Into-Bins Processes.} If the underlying 
network is the complete graph $K_n$, then \ALG\ can be seen as a    
\emph{parallel balls-into-bins}  algorithm~\cite{ACMR98,LW11} with $m=dn$ balls,  each one representing an outgoing-link  
request which must be assigned to one of $n$ bins, corresponding to the 
nodes of the network. In this perspective, our  algorithm    assigns  
each ball to one bin, so that  the maximum load of the bins is at most 
$c d$, for some  constant $c$ and the algorithm terminates in $\bigO(\log n)$ rounds with high probability.    
Several   algorithms have been introduced for this 
problem and the best  algorithms achieve constant maximum load within a 
constant number of rounds by  using $k>1$ random choices at every round 
for each ball~\cite{LW11}. The RAES strategy adopted by our algorithm  
is similar to the one used in the basic version of Algorithm 
\textsc{parallelthreshold} analysed in~\cite{BFLS} by Berenbrink et Al, 
which is in turn a  parallelized version of the scheduling strategy 
studied in~\cite{BKSS13}. They show that the convergence time is $\bigO(n 
\log m)$ when $cd = d + 1$, while our analysis implies that it is $\bigO(\log 
m)$ when $c$ is an absolute constant larger than 1. The maximum 
number of balls accepted by each bin, called the \emph{threshold}, is 
fixed to $\lceil m/n\rceil +1$. They show  this basic version, 
achieving an almost tight maximum load, converges within        $\bigO(n 
\log m)$ rounds, w.h.p. They also conjecture a tight lower bound on the 
convergence time.

\section{Preliminaries and main result}\label{se:prelim}
For an undirected graph $G = (V, E)$,  the \textit{volume} of a subset of nodes
$U \subseteq V$, is $\vol(U) =  \sum_{u\in U}\grado_u$. Notice that when $G$ is
$\Delta$-regular, we have $\vol(U) = \Delta |U|$. Consider two (not necessarily disjoint) subsets $U, W\subseteq V$, we define $e(U, W)$ as the number of edges in $G$ with one endpoint in $U$ and the other in $W$.

\begin{definition}\label{def:expander}
A graph $G = (V, E)$ is an $\epsilon$\textit{-expander} if, for every subset 
$U\subset V$ with $|U| \le n/2$, the 
 number $e(U,V-U)$ of edges in the cut $(U, V-U)$ is at least 
$\epsilon\cdot\vol(U)$.
\end{definition}

\noindent The expansion properties we derive for the subgraph returned by
Algorithm \ALG\ turn out to depend on the spectral gap of the input graph. In
particular, our analysis uses the following ``one-sided'' version of the
\textit{Expander Mixing Lemma}~\cite{LP17}, which establishes a connection
between the second largest eigenvalue of the adjacency matrix of $G$ and its
expansion properties and also holds for bipartite graphs. 

\begin{lemma}\label{le:exp_mixing_lemma}
Assume $G = (V, E)$ is a $\Delta$-regular graph and let $\lambda$ be the second
largest eigenvalue of $G$'s adjacency matrix\footnote{Since $G$ is
$\Delta$-regular, the bounds on $\lambda$ we derive immediately translate into
bounds on the second largest eigenvalue of $G$'s normalized matrix.}.  Let $S$
be any subset of nodes. Then, the number $e(S,S)$ of edges of $G$ with both
endpoints in $S$ is at most
\[
\frac 12 \, \left(\frac{ \Delta  |S|^2 }n \, + \,   \lambda   |S| \right)  \, .
\]
\end{lemma}

\begin{proof}
If $1_S$ is the indicator vector of $S$, then, it holds that
\[
1_S^\intercal A 1_S\, = \, \sum_{v \in V} |N_v(S)| \, = \, 2 \cdot e(S,S) \,, 
\]
where $N_v(S)$ is the set of $v$'s neighbors in $S$   and $e(S,S)$ is   the
number of edges with both end-points in $S$. Observe that the  matrix $A -
\Delta J/n$ (where $J$ is the matrix having all entries set to 1)  has largest
eigenvalue $\lambda$, so we get 
\[
1_S^\intercal ( A - \Delta J/n) 1_S
\, \leq \, \lambda  \| 1_S \|^2 
\, = \, \lambda \cdot |S| \,.
\]
We also notice that $1_S (\Delta J /n) 1_S^\intercal \, = \, \Delta |S|^2/n$. It thus
follows that   
\[
e(S,S) \, \leq \, \frac 12 \, \left(\frac{ \Delta  |S|^2 }n \, + \,   \lambda
|S| \right) \,.
\]
\end{proof}

In the next sections, we analyze the behaviour of Algorithm \ALG\ on 
dense, regular expanders. The algorithm was informally described in the 
introduction, a more formal description is given below.

\begin{algorithm}[H]
\caption{\ALG($G$, $d$, $c$)}\label{alg:raes}
	\begin{algorithmic}[1]
		\State $H:=$ empty directed graph over the node set $V$
		\While{$H$ has nodes of outdegree $<d$}
			\State {\sc Phase 1:}\Comment{$\dout$: current outdegree of $v$ in $H$}		
			\For {each node $v\in V$}
				\State $v\in V$ picks $d-\dout$ neighbors in $G$ uniformly at random
				\State $v$ submits a connection request to each of them
			\EndFor
			\State {\sc Phase 2:}\Comment{$\din$: current indegree of $v$ in $H$}
			\For {each node $v\in V$}
				\If {$v$ received $\leq cd - \din$ connection requests in the 
				previous phase}
					\State $v$ accepts all of them and the corresponding 
					directed links are added to $H$
				\Else
					\State $v$ rejects all connection requests received 
					in   Phase 1
				\EndIf
			\EndFor
		\EndWhile
		\State Replace each directed link by an undirected one\\
		\Return $H$
	\end{algorithmic}
\end{algorithm}

\noindent We next define the class of almost-regular graphs \ALG\ 
stabilizes on w.h.p.

\begin{definition}\label{def:almost-reg}
A graph $G = (V, E)$ is a $(\nreq, c\nreq)$\textit{-almost regular} graph if
the degree $\grado_v$ of any node $v\in V$ is such that $\grado_v \in
\{d,\ldots , (c+1) d\}$.
\end{definition}

\noindent Our main results can be formally stated as follows.
 
\begin{theorem}\label{thm:main-convergence}
For every $d \geq 1$, every $0 < \alpha \le 1$, every $c \geq
2/\alpha$, and
for every $\Delta$-regular graph $G = (V, E)$ with $\Delta = \alpha n$,
the time complexity of \ALG$(G,d,c)$ is $\bigO(n)$ w.h.p.
Moreover, the algorithm  can be implemented in the uniform \gossip\ distributed
model\footnote{In this very-restrictive communication model~\cite{CHKM12,H13},
at every synchronous step, each node can (only) contact a \textit{constant}
number of its neighbors, chosen u.i.r., and exchange two messages (one for each
direction) with each of them.} so that its parallel completion time is
$\bigO(\log n)$ and its overall message complexity is $\bigO(n)$, w.h.p.
\end{theorem}

\begin{theorem}\label{thm:main-expanders}
A sufficiently small constant $\epsilon > 0$ exists such that, for 
any constants $d\geq 44$ and  
$0 < \alpha \le 1$, for any sufficiently large $c$\footnote{We didn't try to optimize the constants in our analysis, which shows that $c \geq \max \{(\frac 2{\alpha})^2, 10e^{10d} \}$ suffices.}, 
and every $\Delta$-regular graph $G = (V, E)$ with $\Delta =
\alpha n$ and second largest eigenvalue of the adjacency matrix\footnote{I.e., $G$ is a sufficiently good
expander. Also note that, equivalently, we are imposing that the second largest
eigenvalue of $G$'s normalized adjacency matrix be at most $\epsilon\alpha^2$.}
$\lambda\leq\epsilon\alpha^2\Delta$,
\ALG$(G,d,c)$
returns a $(\nreq, c \nreq)$-almost regular $\epsilon$-expander $H =
(V,A)$, w.h.p.
\end{theorem}

The proof of Theorem~\ref{thm:main-convergence} is given in
Section~\ref{subsec:expconvtime}, while the proof of
Theorem~\ref{thm:main-expanders}, which is our main technical contribution, is
described in Section~\ref{sub:expexp}.

\section{Proof of Theorem~\ref{thm:main-convergence}} \label{subsec:expconvtime}
Throughout this section, we consider a $\Delta$-regular graph $G = (V,E)$ with
$\Delta = \alpha n$ for some arbitrary constant $0 < \alpha < 1$. We analyze
the execution of Algorithm~\ALG\ on input $G$ for any constants $d \geq 1$ and
$c > 1/\alpha$.
  
Recall that, according to the process defined by \ALG, each node $v$ asks for
$d$ \emph{link requests} to its neighbors and has $cd$ \emph{slots} to
accomodate link requests from its neighbors. 

We first provide a simple proof that \ALG\ on input $G=(V,E)$ terminates within
a logarithmic number of rounds\footnote{Notice that the meaning of \emph{round}
here is exactly that defined in the pseudocode of \ALG.}, w.h.p. 
 
\begin{lemma} \label{lem:termination}
For every $d \geq 1$, every $c > 1/\alpha$, and every $\beta > 1$, \ALG$(G,d,c)$
completes the task within $\beta \log(n) / \log(\alpha c)$ rounds, with probability at
least $1 - d/n^{\beta-1}$.
\end{lemma}

\begin{proof} 
Let us fix an arbitrary ordering of the $nd$ required links and, for $i = 1,
\dots, nd$, let $X_i^{(t)}$ be the binary random variable taking value $1$ if
link $i$ is settled at the end of round $t$ and $0$ otherwise. 

First note that, since a link is settled at some round $t$ if it was
already settled at previous round $t-1$, it holds that
\begin{equation}\label{eq:linkrequesti}
\Prob{}{X_i^{(t)} = 0} = \Prob{}{X_i^{(t)} = 0 \;|\; X_i^{(t-1)} = 0}
\Prob{}{X_i^{(t-1)} = 0} \, .
\end{equation}
Let us name $\bY^{(t)}_{-i} = \left(Y_1^{(t)}, \dots, Y_{i-1}^{(t)},
Y_{i+1}^{(t)}, \dots, Y_{nd}^{(t)}\right)$ the random vector where, for each $j
\neq i$, random variable $Y_j^{(t)}$ indicates the destination node of 
link $j$ at round $t$. Observe that for every vector $\by_{-i} = (y_1, \dots,
y_{i-1}, y_{i+1}, \dots, y_{nd}) \in V^{nd -1}$ it holds that
\begin{equation}\label{eq:linkrequestscond}
\Prob{}{X_i^{(t)} = 0 \;|\; X_i^{(t-1)} = 0, \, \bY_{-i}^{(t)} = \by_{-i}} 
\leq \frac{1}{\alpha c} \, .
\end{equation}
Indeed, given any $\by_{-i} \in V^{nd -1}$, there are always at most $nd/(cd) =
n/c$ nodes with $cd$ or more incoming link requests. Hence, among the $\alpha n$
neighbors of the node asking link $i$, at least $(\alpha - 1 /c)n$ have less
than $cd$ incoming requests. Hence, the probability that link $i$ settles is
at least $1 - 1/(\alpha c)$. Since~\eqref{eq:linkrequestscond} holds for any choice of 
 $\by_{-i} \in V^{nd -1}$,
  we get that
$\Prob{}{X_i^{(t)} = 0 \;|\; X_i^{(t-1)} = 0} \leq 1/(\alpha c)$ and thus 
from~\eqref{eq:linkrequesti} we have that
$\Prob{}{X_i^{(t)} = 0} \leqslant 1/(\alpha c)^t$.
The thesis then follows from a union bound over all the $nd$ links and 
from the fact that $t \geqslant \beta \log(n) / \log (\alpha c)$.
\end{proof}

\smallskip \noindent \textbf{Remark.} The first proof we gave for the above
lemma was based on a simple compression argument~\cite{MMR17}. We describe it
in Appendix~\ref{app:time} since it can be used by the reader as a ``warm-up''
for the more difficult analysis given in Section~\ref{sub:expexp}.

\smallskip The time complexity of Algorithm \ALG\ is asymptotically bounded by
the total number of link requests produced by its execution on graph $G =
(V,E)$.  Lemma~\ref{lem:termination} easily implies that this number is
$\bigO(dn \log n)$, w.h.p. In the next lemma we prove a tight $\bigO(nd)$
bound.
 
\begin{lemma} \label{lem:time}
For every constants $d \geq 1$ and $c > 2 / \alpha$, the total number of link
requests made by \ALG$(G,d,c)$   (and thus the time complexity) is
$\Theta(n)$, w.h.p.
\end{lemma}

\begin{proof} 
Let us fix an arbitrary ordering of the $nd$ required links and, for $i = 1,
\dots, nd$, let $Z_i^{(t)}$ be the binary random variable taking value $1$ if
link $i$ is not yet settled at the beginning of round $t$ and $0$ otherwise. 
The random variable indicating the total number of link requests produced by 
the algorithm can thus be written as
\[
Z = \sum_{t = 0}^\infty \sum_{i = 1}^{nd} Z_i^{(t)} \, .
\]
Proceeding as in the proof of Lemma~\ref{lem:termination}, it is easy to see
that for every $t \in \mathbb{N}$ it holds that
\[
\Expec{}{\sum_{i = 1}^{nd} Z_i^{(t)}} \leqslant \frac{nd}{(\alpha c)^t} \,.
\]
Hence, the total expected number of link requests is $\Expec{}{Z} \leq
\frac{\alpha c}{\alpha c - 1} \, nd$. 

In order to prove that $Z = \bigO(nd)$ w.h.p., we first show that whenever the
number of unsettled links is above $nd / \log n$, it decreses by a constant
factor, w.h.p. Formally, for any $k \geqslant nd / \log n$, we derive the 
following inequality

\begin{equation}
\Prob{}{\sum_{i=1}^{nd} Z_{i}^{(t)} >  \frac{k}{\alpha c / 2}  \; \middle| \;
\sum_{i=1}^{nd} Z_{i}^{(t-1)} = k}
\leq e^{-\frac{k}{2 \alpha^2 c^4 d^2}} \, .
\label{eq:concentration}
\end{equation}

Notice that random variables $Z_{i}^{(t)}$ conditional on the graph formed by
the links settled at the end of round $t-1$ are not independent, so we cannot
use a standard Chernoff bound. However, we can use the \emph{method of bounded
differences}~\cite[Corollary 5.2]{DP09} (see Theorem~\ref{thm:mobd} in
Appendix~\ref{app:tools}), since the sum of the $Z_i^{(t)}$ conditional on the
graph of the $nd-k$ settled links at the end of the previous round can be
written as a $2cd$-Lipschitz function of the independent $k$ random variables
indicating the link requests at round $t$. 

In more details, we name $u^{(t-1)}$ the set of $k$ unsettled links at the end
of round $t-1$ and consider random variables $\{Y_{i}\} _{i \in u^{(t-1)}}$,
each of them returning the node-destination index that the non-assigned link
request $i$ tries to connect to. Observe that $Y_i$'s are mutually independent
and, moreover, the sum in (\ref{eq:concentration}) can be written as a
deterministic function of them: $\sum_{i=1}^{nd} Z_{i}^{(t)} =
f\left(Y_{i_1},\dots, Y_{i_k}\right)$.

Moreover, this function is $2cd$-Lipschitz w.r.t. its arguments: If we change
one of the arguments $Y_{i}$, we are moving a request $i$ from a node $v_{1}$
to a node $v_{2}$. The largest impact this can have on $\sum_{i=1}^{nd}
Z_{i}^{(t)}$ is that the response for each of all the link requests sent to
$v_{2}$ changes. However, if this number was already larger than $cd$, then the
moving of link request $i$ would not have any impact. This means that, in the
worst-case, at most $cd$ link requests trying to connect to $v_{2}$ switch from
assigned to non-assigned.
At the same time, a symmetric argument holds for the link requests
trying to connect to $v_{1}$. In formulas, for all vectors of nodes
$\left(v_{i_1},\dots,v_{i_j},\dots,v_{i_k}\right)$ and
$\left(v_{i_1},\dots,v_{i_j}^{\star},\dots,v_{i_k}\right)$ differing
only on a single entry $i_j$, it holds that 
\[
\left|f\left(v_{i_1},\dots,v_{i_j},\dots,v_{i_k}\right)-f\left(v_{i_1},\dots,v_{i_j}^{\star},\dots,v_{i_k}\right)\right|\leq 2cd \, .
\]
Therefore, by applying Corollary 5.2 in \cite{DP09} (see also
Theorem~\ref{thm:mobd} in Appendix~\ref{app:tools}), with $\mu \leq M = k/(\alpha
c)$ and $\beta_j = 2cd$ for all $j = 1, \dots, k$ we get
(\ref{eq:concentration}).

From (\ref{eq:concentration}) and the chain rule, it follows that, for $T =
\bigO\left( \frac{\log\log n}{\log\left(\alpha c/2 \right)}\right)$ rounds, the
number of unassigned link requests decreases by a factor $\alpha c / 2 > 1$ at
each round, w.h.p., until it becomes smaller than $nd/\log n$. These rounds
thus account for $nd \sum_{t=0}^{T}\left(\frac{2}{\alpha c} \right)^{t} =
\bigO(nd)$ connection requests, w.h.p. Then, from Lemma~\ref{lem:termination}
it follows that the remaining $\frac{nd}{\log n}$ link requests are assigned
within $\bigO(\log n)$ rounds, w.h.p., thus accounting for at most further
$\bigO(n d)$ additional link requests, w.h.p.  
\end{proof}

\subsubsection*{Distributed implementation}
As one can easily verify from its pseudocode, Algorithm~\ALG\ is designed to
work over any synchronous parallel distributed model  where  the  nodes of the
input graph $G=(V,E)$ are the local computing units which can communicate via
the bidirectional links defined by the set of edges $E$. We remark that, at
every round, each node contacts (i.e. sends link requests to) only a constant
number of its neighbors. It thus follows that \ALG\ induces a decentralized
protocol that can be implemented on the  communication-constrained
\emph{uniform} \gossip\ model~\cite{CHKM12,H13}. Notice that the protocol
does not require any global labeling of the nodes, rather, it requires
that each node knows some local labeling of its bi-directional ports.   

In this setting, Lemma~\ref{lem:termination} easily implies that every node
completes all of its tasks within $\Theta(\log n)$ rounds, w.h.p.

As for communication complexity, we observe that all the point-to-point
communications made by the protocol can be encoded with 1-bit messages
(accept$/$reject the link request). Moreover, Lemma~\ref{lem:time} implies that
the overall number of links requests (and thus of exchanged messages) is
w.h.p. $\Theta(dn)$, which is clearly a tight bound for this task. 

Finally,  we notice that if nodes
know an upper bound $n'$ on $n$, since $G$ is regular, then
they can locally derive a sufficiently good lower bound of $\alpha$, i.e.,
$\alpha' = \alpha/\poly(n)$. Then, by Lemma~\ref{lem:termination}, after round
$T=2 \log(n') / \log(\alpha' c)$, every node can decide to stop any action (so
it terminates) and it will be aware that the protocol has completed the global
task, w.h.p.

\section{Proof of Theorem~\ref{thm:main-expanders}} \label{sub:expexp}
In the previous subsection, we showed that, after $T = \bigO(\log n)$ rounds,
Algorithm~\ALG\ stabilizes to a subgraph $H=(V,E_H)$ of the input graph
$G=(V,E)$ that turns out to be a $(\nreq, c\nreq)$-almost regular graph. In
this Section, we provide the proof of Theorem~\ref{thm:main-expanders}: we
indeed show that if $G$ is an expander then  $H=(V,E_H)$ turns out to be also an  
expander, w.h.p.  The proof proceeds by showing that $\Prob{}{(\text{\ALG\ completes in $T$ rounds})\wedge (\text{$H$ is not an $\epsilon$-expander})} = \bigO(\frac{1}{n^\gamma})$ for a constant $\gamma$. Combined with Theorem \ref{thm:main-convergence}, this proves Theorem \ref{thm:main-expanders}, with $T = \bigO(\log n)$.

In the next subsection we sketch the main  arguments we use in the proof. Then in
the successive subsections we provide the  detailed proof.

\subsection{Overview of the proof}\label{ssec:overview}
The probability distribution of the links yielded by \ALG, and the ways in
which the links are correlated, are very difficult to analyze. In order to cope
with such technical issue, we prove Theorem~\ref{thm:main-expanders} by using a
compression argument: we show that the random choices of the algorithm that
lead to a non-expading graph can be non-trivially compressed, and hence have
low probability. 

We think of each node as having access to a sequence of $Td\log \Delta$ random
bits and the protocol as being deterministic as a function of these $n$ local
sequences of random bits (see Fig.~\ref{fig:uncompr}). We will show that any
sequence of $nTd \log \Delta$ bits leading the protocol to stabilize within $T$
rounds to a non-expanding graph can be losslessly described using $nTd \log
\Delta - \Omega(\log n)$ bits. This will prove that the protocol stabilizes to
an expanding graph with high probability.

Let $R \in \{0,1\}^{nTd\log \Delta}$ be a bit string that leads to a
non-expanding set, i.e., a set $S$, with size $|S| = s \leq n/2$, having at
most $\epsilon |S| d$ outgoing links in $H$. The compression of such a bit
string $R$ is based on two main ideas. Since the number of links in $E_H$ with both
endpoints inside $S$ is large, the first main idea is to use less than $\log
\Delta$ bits to encode the destination of each accepted requests originated
from nodes in $S$ whenever this destination belongs to $S$. The second main
idea is to encode the destinations of rejected requests with less than $\log
\Delta$ bits. Indeed, roughly speaking, for each link request that gets
rejected at some round, there are at least further $cd$ link requests (in the
current or previous rounds) towards the same ``bad'' destination. Since there
are a total of $dn$ requests that need to be accepted and a total of $cdn$
available accepting slots, the number of such ``bad'' destinations needs to be
small, thus the destinations of rejected requests may be compressed.

For clarity sake, we think of the original available randomness $R \in \{0,
1\}^{nTd\log\Delta}$ organized as an $n \times Td$ matrix $\sM$, where each
entry is a block of $\log \Delta$ bits (see Fig.~\ref{fig:uncompr}). The
compressed counterpart of $\sM$, denoted as $\sC$ (see Fig.~\ref{fig:compr}),
consists of three tables (a detailed description of each table can be found in
the next subsection).

In order to implement the first idea, for each node $v \in S$ we need to
identify, in its row of the uncompressed representation $\sM$, which slots of
$\log \Delta$ bits refer to destinations in $S$ of accepted link requests. Notice
that this cannot be done naively, indeed even just identifying the set of slots
of accepted requests would naively require $\log\binom{Td}{d}$ bits.  However,
we can first encode the number $\ell_v$ of used slots, with a prefix free
encoding\footnote{See, e.g., \textit{Elias $\delta$-coding}
(\url{https://en.wikipedia.org/wiki/Elias_delta_coding}).} requiring at most
$\log \ell_v$ bits and then the set of $d$ slots referring to accepted requests
by using $\log \binom{\ell_v}{d}$ bits (see Fig.~\ref{fig:compr}: Field~1 in
Table~2). While for some ``unlucky'' nodes $\ell_v$ can be large, the overall
amortized number of bits $\sum_{v \in S} \left[ \bigO(\log \ell_v) + \log
\binom{\ell_v}{d} \right]$ turns out smaller than $s \log\binom{Td}{d}$.  Once
we have identified the set of accepted requests, we can identify the set of
those referring to destinations in $S$ (see Fig.~\ref{fig:compr}: Field~2 in
Table~2) and, finally, we can encode each of those destinations by using $\log
[(1-\delta_v) \Delta]$ bits instead of $\log \Delta$ bits, where $\delta_v$ is
the fraction of neighbors of $v$ outside $S$ (see Fig.~\ref{fig:compr}: Field~3
in Table~2). Notice that we can identify which requests of each node end up
inside and outside $S$ by encoding the set $S$ itself, once and for all, using
$\bigO(\log s) + \log\binom{n}{s}$ bits (see Fig.~\ref{fig:compr}: Table~1).

In order to implement the second main idea, for each node $v \in S$ we need to
identify the destinations of its $\ell_v - d$ rejected link requests. Notice
that each rejected request ends up on a node, say $w$, receiving at least
further $cd$ requests. Those further requests include requests
\textit{accepted} by $w$ in some previous round and requests \textit{rejected}
by $w$ in the current round. We exploit that property to reduce the number of
bits used to encode such destinations: roughly speaking, at each round $t$ we
distinguish between \textit{semi-saturated} and \textit{critical} nodes. We
call semi-saturated at round $t$ a node that already accepted at least $cd/2$
requests up to round $t-1$. Notice that (i) the number of semi-saturated nodes
can never exceed $2n/c$ and (ii) we already know the set of semi-saturated
nodes at round $t$, if we know the accepted requests of all nodes up to round
$t-1$.  Hence, we can encode each request to a semi-saturated node by using
only $\log(2n/c)$ bit (notice that this is smaller than $\log \Delta$ whenever
$c > 2 / \alpha$). In order to distinguish which ones of the $\ell_v -d$
rejected destinations refer to semi-saturated nodes and which ones refer to
critical nodes we use further $\ell_v - d$ bits. Finally, for critical nodes
(i.e., destinations of rejected requests that are not semi-saturated) we first
encode once and for all the set of such nodes at each round, using $\bigO(\log
c_t) + \log \binom{n}{c_t}$ for each round $t$, so that we can encode the
destination of a rejected request toward a critical node at round $t$ using
only $\log c_t$ bits.

Summing up all the contributions involved (see
Section~\ref{ssec:rateofcompression} for all the details) we end up encoding
a string $R \in \{0,1\}^{nTd\log \Delta}$ leading to a non-expanding set with
a bit string of length $nTd\log \Delta - \Omega(\log n)$. Thus the overall
number of bit strings leading to non-expanding sets is at most an
$\bigO(n^{-c})$ fraction of all the bit strings, for some $c > 0$.

\begin{figure}
\centering
\includegraphics[width=5cm]{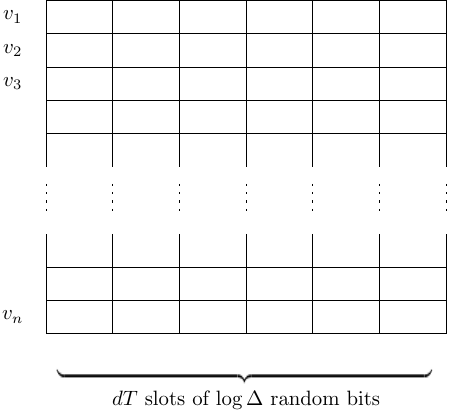}
\caption{Uncompressed representation $\mathcal{M}$ of $R$} \label{fig:uncompr}
\end{figure}

\begin{figure}
\centering
\includegraphics[width=12cm]{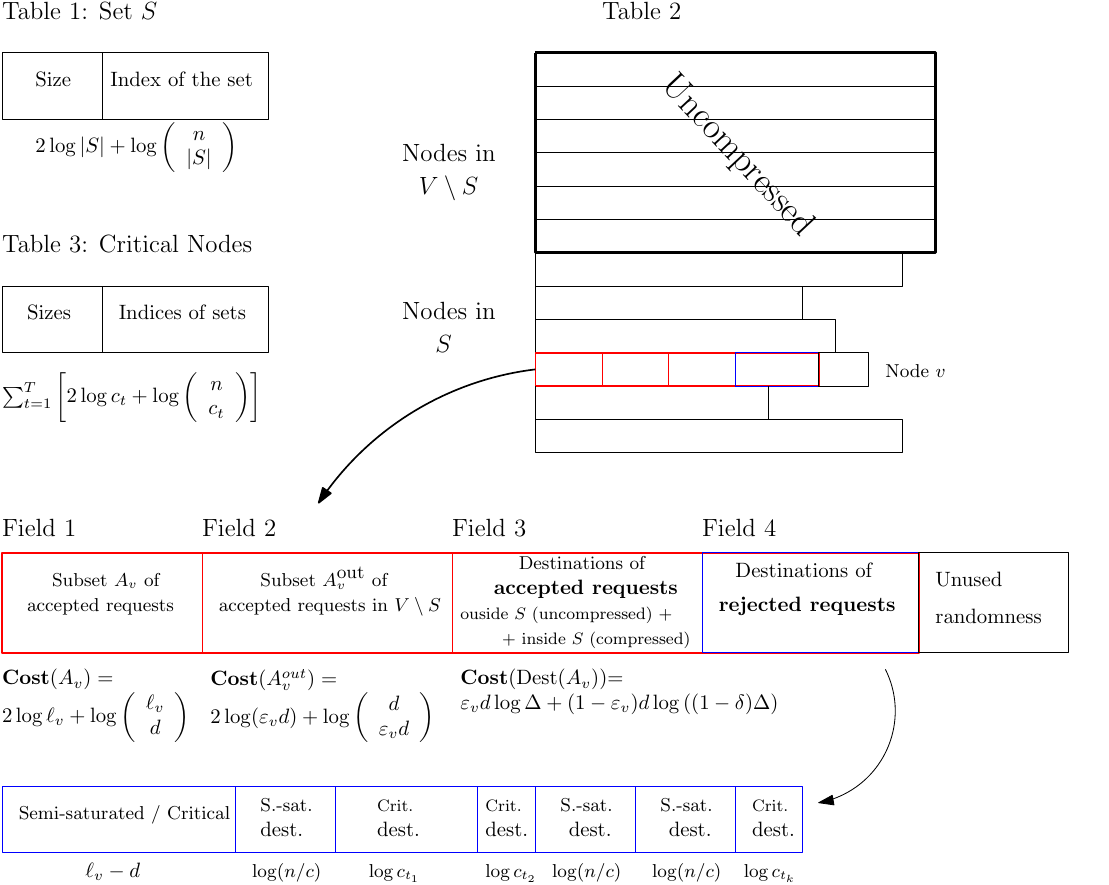}
\caption{Compressed representation $\mathcal{C}$ of $R$} \label{fig:compr}
\end{figure}

%

\subsection{The compressed representation: Full description}
\label{ssec:description}
We use the following notation throughout the remainder of the paper.  For a
node $v\in S$, we denote by $\out_v$ the fraction of $v$'s edges in $E$ that
have an end-point in $V - S$, i.e., 

\begin{equation} \label{eq:defeps&delta}
\out_v \cdot \Delta \ = \   e_G(v, V - S)  \, \mbox { and } 
\, \out \, =\,  \frac{1}{s}\sum_{v \in S}\out_v \, . 
\end{equation} 
We also denote by $\freq_v$ the fraction of $v$'s accepted link requests (so
edges of subgraph $H$) with end-points in $V - S$, i.e., 
\begin{equation}\label{eq:fractions}
\freq_v \cdot d \, = \  e_H(v, V-S)  \, \mbox{ and } 
\, \freq = \frac{1}{s}\sum_{v\in S}\freq_v  \, . 
\end{equation}

\medskip In the paragraphs that follow, we describe how the evolution of the
protocol is encoded in the presence of a non-expanding subset $S$ (with $|S| =
s\le n/2$ without loss of generality).

In the remainder, we repeatedly use the following facts:
\begin{itemize}
\item {\bf Node numbering:} when representing destinations of link requests 
submitted by nodes of the network, we can use the fact that the encoding 
and decoding algorithms have full knowledge of the underlying graph. In 
particular, we assume a total ordering of the nodes is defined, so that 
a node $u$ is simply specified by an integer in $\{1, ... , n\}$, 
denoting $u$'s position in this ordering. At the same time, we can use 
a local numbering to represent the neighbors of a given node $v$. For 
example, if $v$'s neighbors are the nodes $\{2, 5, 8\}$ with respect to 
the global ordering, node $5$ can be represented as $2$ with respect to 
$v$, i.e., the second neighbor of $v$ with respect to the global ordering. 
\item {\bf Subset encoding:} given the set $[k]$ of the first $k$ 
integers, we represent any subset $S\subset [n]$ by its position $i$ in 
the lexicographic order of all subsets of $[n]$ of size $|S|$. In order 
to completely specify $S$, we separately encode its size in a prefix-free way 
using $2\log|S|$ bits, and its position $i$ in the lexicographic order 
using $\log\binom{n}{|S|}$ bits. 
\end{itemize}

We next discuss the compressed encoding we use. We remark that, as argued in \cite[Section 7]{MMR17}, we can avoid taking ceilings in the expressions which measure the number of bits necessary for the encoding. 

\paragraph{Unused randomness.} For every node $v$, we have enough randomness to
describe exactly $dT$ choices.  If $v$ completes its execution of the protocol
after performing $\ell_v$ requests, the remaining randomness (corresponding to
$dT - \ell_v$ requests) is not used. This unused randomness is both present in
the uncompressed representation $\sM$ and in its compressed counterpart $\sC$
and is represented as is, thus corresponding to $\sum_{v\in V} (dT -
\ell_v)\log\Delta$ bits in both $\sM$ and $\sC$.

\paragraph{Table 1: The set $S$.} We represent $S$ in $\sC$ by writing the
number $s:= |S|$ in a prefix-free way using $2\log s$ bits, and then writing
the number $k$ such that $S$ is the $k$-th set of size $s$ in lexicographic
order, which takes $\log {n \choose s}$ bits. Using   prefix $\delta$-codes, in
total, the cost to encode $S$ is
\begin{equation}\label{reps}
\bin(S) = 2\log s + \log {n \choose s} \, .
\end{equation}

\paragraph{Table 2, upper part: Randomness of nodes in $V-S$.} We represent the
randomness of nodes in $V-S$ as it is, with no gain or loss. Notice that,
thanks to  Table 1, a decoder can infer that the first $(n-s)$ rows of Table 2
(see Fig \ref{fig:compr}) describe the executions of every node in $V - S$.
This is a fixed-length encoding formed by $(n-s)$ rows, each consisting of $dT$
blocks of $\log\Delta$ bits:    hence a decoder  knows where the lower portion
of $\sC$ encoding the executions of nodes in $S$ begins.

\smallskip For every node $v \in S$, the lower part of Table 2 contains a
\emph{variable-length} row, in turn consisting of a set of consecutive
\emph{fields}, which encode the following information.

\paragraph{Table~2, Field~1: Subset $A_v$ of requested links originating from $v$ that
are accepted.} This field consists of two parts. In the first part we  write,
in a prefix free way, the number $\ell_v$, using $2\log \ell_v$ bits. As a
second part of this field, we specify the subset of the $d$ accepted link requests among
the $\ell_v$ submitted by $v$.\footnote{Recall that we are encoding executions
of the algorithm that terminate within $T$ rounds.} To this purpose, we again
encode the integer $i$, such that the $d$ accepted requests correspond to the
$i$-th subset of $\{1, \ldots , \ell_v\}$ of size $d$, in lexicographic order.
The overall cost incurred for this field is thus 
\begin{equation}
\label{repaccs}
\bin(A_S) = \sum_{v\in S}\bin(A_v) 
= \sum_{v\in S} 2 \log\ell_v + \log {\ell_v \choose d} \, .
\end{equation}

\noindent \emph{Remark:} note that this field allows to iteratively infer the
round in which each request was submitted by $v$. Also notice that the subset
of rejected link requests originating from $v$ can be derived as the complement
of subset $A_v$.

\paragraph{Table~2, Field~2: Subset $A_v^{\mbox{out}} \subseteq A_v$ of
accepted links originating from $v$ to $V-S$.} For a  node $v\in S$, recall
that $\epsilon_v d$ is the number of outgoing accepted  links from $v$ into
$V-S$. In this field, we  encode the  subset $A_v^{\mbox{out}}$    of  such
accepted links, using the same encoding used for subset $A_v$ in the first
field. We can thus recover the relative positions of such accepted requests in
the overall sequence of the $\ell_v$ requests made by $v$. In total, this cost
is
\begin{equation}
\label{repcut}
\bin(A_S^{\mbox{out}}) 
= \sum_{v\in S}\bin(A_v^{\mbox{out}}) 
= \sum_{v \in S} 2 \log (\epsilon_v d) + \log {d \choose \epsilon_v d } \, .
\end{equation}

\noindent \emph{Remark:} note that the encoding of $A_v^{\mbox{out}}$ is
relative to subset $A_v$. For example, if $d = 4$ and $A_v = \{1, 3, 5, 6\}$,
we would know that the first, third, fifth and sixth requests placed by $v$
were accepted. Moreover, if $A_v^{\mbox{out}} = \{2, 3\}$, we would know that
out of these, the second and third (i.e., the third and fifth request out of
the $\ell_v$ submitted by $v$) had destination in $V - S$.

\smallskip
Note that the boundary between the first field and the second  one above  is 
uniquely determined by the value of $\ell_v$, which is encoded in a 
prefix-free way. The same holds for the second field.

\paragraph{Table~2, Field~3: Destinations of accepted links originating from
$v$.} This field consists of two parts. In the first part, we  represent
accepted links with  destinations in $V-S$ as they are (i.e., using
$\log\Delta$ bits), with no gain or loss. In the second part, we represent
destinations of accepted links in $S$ using $\log((1 - \out_v)\Delta)$ bits
instead of $\log\Delta$.  Overall, the cost we incur is
\begin{equation}\label{reprandacc}
\bin(\dest(A_S)) = \sum_{v\in S}(1 - \freq_v)\nreq\log ((1 - \out_v)\Delta) 
+  \freq_v\nreq\log\Delta \, .
\end{equation}

\noindent \emph{Remark.} Note that here we are using a local numbering for
neighbors of $v$ that belong to $S$. Moreover, thanks to the information
encoded in the previous fields (i.e. the size of $A_v^{\mbox{out}}$ and that of
$A_v$)  we can use a  standard block  code for both the above parts since we
know exactly their respective lengths.

\paragraph{Table~2, Field~4: Destinations of rejected requests originating from
$v$.} We finally compress the encoding of the destinations of rejected requests. 
In order to do so, we first introduce the following notions.
 
\begin{definition}[Semi-saturated and Critical Nodes.]
\label{def:critical}
We call a node $w$ 
\begin{itemize}
\item {\em semi-saturated at round $t$}, if the number of
accepted incoming links up to round $t-1$, plus the number of requested
links at round $t$ originating from nodes in $V-S$ is at least $cd/2$. 
\item {\em critical at round $t$}, if it is \emph{not} semi-saturated at round
$t$ but it has more than $cd$ links (accepted or requested) at round
$t$ (note that this implies that $w$ received more than $cd/2$
requests from $S$ at round $t$).
\end{itemize}
\end{definition}

We will make use of the following facts.

\begin{lemma}\label{lemma:saturated}
For every round $t$, it holds that: 
\begin{itemize}
\item The number of semi-saturated
nodes is at most $2n/c$.
\item The number of critical
nodes is at most  $n/c$.
\end{itemize}
\end{lemma}
\begin{proof}
Consider a node that is semi-saturated at round $t$. This node was the
recipient of at least $cd/2$ link requests, that it either accepted {\em
before} round $t$, or it received in round $t$.  Since, from the definition of
\ALG, for every node, the overall number of its link requests that are accepted
within round $t-1$, plus the number of link requests it issues at round $t$
cannot exceed $d$, we have a total of at most $dn$ such requests over the
entire network. This immediately implies that the number of semi-saturated
nodes at round $t$ cannot exceed $2n/c$.  The argument for the number of
critical nodes at round $t$ proceeds along the same lines and is omitted for
the sake of brevity.
\end{proof}

\smallskip\noindent
In what follows we represent the subsets of semi-saturated and critical nodes.

\begin{itemize}

\item \textbf{The subset of semi-saturated nodes at each step.} From its
definition, the set of semi-saturated nodes needs not be represented
explicitely. In fact, for every round $t$, this set is uniquely determined by
the evolution of the protocol (and thus by the corresponding portions of our
Tables) up to round $t-1$ and by link requests issued by nodes in $V - S$ at
round $t$, whose randomness is represented as it is (see Table 2, upper
portion).

\item \textbf{Table 3: The subset of critical nodes at each step.} We represent
the subsets of critical nodes in each round explicitely.  Let $C_t$ be the set
of critical nodes at round $t$ and let $c_t:= |C_t|$.  We represent all such
sets in a separate table (see Table~3 in Figure~\ref{fig:compr}). This table
consists of two fields. The first is the sequence of the critical set sizes,
encoded in a prefix-free way. The second field is the sequence of the integers
representing $C_t\subset V$, for $t=1, \ldots, T$. Note that, the length of the
field encoding $C_t$ is completely determined once we know $c_t$. Overall,
encoding information in Table~3 has cost 
\begin{equation}
\label{repcrit}
\bin{(C)} = \sum_{t=1}^T 2\log c_t + \log { n \choose c_t }  \, .
\end{equation}

\end{itemize}

Given this premise, this field consist of two parts. The first part is a
sequence of exactly $\ell_v - d$ bits. The $i$-th such bit specifies whether
the destination of the $i$-th rejected request was a semi-saturated or critical
node in the round in which the request was issued.  The second part of the
field is simply the sequence of destinations of rejected requests, encoded in
compressed form thanks to Lemma~\ref{lemma:saturated}.  Specifically, for each
round $t$, we represent each rejected connection toward to a critical node
using $\log c_t$ bits (recall that we explicitely represent $C_t$), and each
other rejected connection, which necessarily goes to a semi-saturated node,
using $\log (2n/c)$ bits. 

To compute the corresponding cost of representing destinations of 
rejected requests, let $\rc_t(v)$ be the number of
rejected requests from $v$ to critical nodes in round $t$, and
let $\rss(v)$ be the overall number of rejected connection requests from $v$ to
semi-saturated nodes, over the entire process. Then, the overall 
cost of encoding the destinations of rejected requests from $v$ is
\begin{equation}
 \label{reprandrej}  
 \bin(\dest(\sRJ)) = (\ell_v - d ) +  
 \rss(v) \cdot \log \frac{2n}{c} +   \sum_{t=1}^T \rc_t(v) \cdot  
 \log c_t \ \, .
\end{equation}

Observe that the additive term $(\ell_v - d )$ in the equation above 
corresponds to the aforementioned first part of the field. 

\subsection{Decoding algorithm}
We  show correctness of our
encoding, discussing how the entire evolution of the protocol can be recovered
from its compressed encoding without loss of information.
Before describing this decoding algorithm, it is useful to define, for the
remainder of this section, the notion of \emph{state of \ALG's 
execution} at round $t$. 

\begin{definition}\label{def:state}
The state $\state_t$ of \ALG's execution at time $t$ is a vector, whose 
component $\state_t(v)$ is the ordered sequence of the destinations of all 
link requests issued by $v$ in round $t$.
\end{definition}
We note that knowledge of $\{\state_1,\ldots ,\state_t\}$ allows to fully characterize the 
evolution of the process up to round $t$. In particular, for every round $i = 1,\ldots , t$, 
we can tell exactly which requests were accepted and which were 
rejected in that round.

\paragraph{Further notation used in this subsection.} For a node $v$ 
and a round $t$, we define by $x_t(v)$ and $a_t(v)$ respectively the 
overall number of link requests submitted by $v$ in round $t$ and the 
number of those that were accepted. We let $x_{\le t}(v) = 
\sum_{i=1}^tx_i(v)$ and $a_{\le t}(v) = 
\sum_{i=1}^ta_i(v)$ for conciseness (note that $a_t(v)\le x_t(v)$ and 
$a_{\le t}(v)\le x_{\le t}(v)$ by definition). 
For every $v\in V$, we denote by $\dest_t(v)$ the set of 
destinations of requests issued by $v$ in round $t$. We denote by $\SST_t$ 
and $C_t$ respectively the subsets of semi-saturated and critical nodes in round $t$.

\smallskip
We next outline the main steps of a decoding algorithm $\DEC(G, \sC)$. The
algorithm takes as input the underlying graph $G$ and the compressed encoding
$\sC$ and it returns the evolution of \ALG\ over the at most $T$ steps of its
execution. More precisely, for every $t$, $\DEC(G, \sC)$ returns a special
symbol $\emptyset$ if \ALG\ completed its execution before time $t$. Otherwise,
$\DEC(G, \sC)$ returns $\state_t$, i.e., for every $v$, the sequence of 
requests issued by $v$ in round $t$. Note that this is enough to recover $\sM$,
since unused randomness is represented as is both in $\sM$ and $\sC$. 
In particular we show how, given $G$, $\sC$ and $\{\state_1, \ldots 
,\state_{t-1}\}$, it is possible to recover $\state_t$.\footnote{Note that 
$\state_0$ simply contains an empty request sequence for every $v$.}
The main steps of the algorithm are summarized as Algorithm \ref{alg:dec} below, while details on how each piece of information can be recovered from $\sC$ have been discussed in Section \ref{ssec:description}. This is enough to prove that the compressed encoding is lossless. 

\begin{algorithm}[H]
\caption{\DEC($G$, $\sC$)}\label{alg:dec}
	\begin{algorithmic}[1]
		\State Identify $S$ from Table 1
		\For{$t = 1,\ldots , T$}
		    \State Use $\{\state_1,\ldots ,\state_{t-1}\}$ to compute $x_{\le t-1}(v)$ and $a_{\le t-1}(v)$, for every $v\in V$
		    \If {$a_{\le t-1}(v) = d$ for every $v\in V$}
		        return $\emptyset$
		    \EndIf
		    \For{$v\in V - S$}
		        \State Look up $v$'s row in Table 2, using $x_{\le t-1}(v)$ and $a_{\le t-1}(v)$ to identify the set $\dest_t(v)$ of the destinations of the $d - a_{\le t-1}(v)$ requests that were submitted by $v$ in round $t$
		        \State Use $\{\state_1,\ldots ,\state_{t-1}\}$ and $\dest_t(V-S)$ (the latter computed in the previous step) to identify the subset $\SST_t$ of semi-saturated nodes in round $t$
		        \State Use Table 3 to identify the subset $C_t$ of critical nodes in 
		        round $t$
		    \EndFor
		    \For {$v\in S$}
		        \State Use Field 1 of $v$'s row in $\sC$ to identify the subset of $v$'s accepted requests that were submitted in round $t$ and compute their number $a_t(v)$
		        \State Use information collected in the previous step, Field 2 and Field 	3 to identify the destinations of accepted requests submitted by $v$ in round $t$
		        \State Use Field 4 and $\SST_t$ and $C_t$ computed above to identify the destinations of rejected requests submitted in round $t$
		    \EndFor
		\EndFor
		\State return $\state_t$
	\end{algorithmic}
\end{algorithm}

\subsection{Rate of compression}\label{ssec:rateofcompression}
In this subsection, we show that, if $R$ represents an execution terminating
and returning a {\em non expanding} graph $H$, the corresponding encoding
according to the scheme presented in the previous section uses $ndT\log\Delta -
\Omega(\log n)$. In more detail, we apply our encoding scheme described in the
previous subsection to any subset $S\subset V$ that is not an
``$\freq$-expander'' in the graph $H$ returned by \ALG.

The analysis of the achieved compression rate proceeds by carefully bounding
the costs of the compressed representation $R'$ and comparing them with their
counterparts in the uncompressed representation $R$. We first show that the
additive costs (with respect to $R$) of representing the non-expanding set $S$
(see Table~1 of  Fig.~\ref{fig:compr} and~\eqref{reps}) and the subsets $A_S^{\mbox{out}} $ of
accepted requests with destinations in $V - S$ (see Field~2 in Table~2
and~\eqref{repcut}) are more than compensated by the compression achieved in
the representation of accepted requests with destinations in $S$ (see Field~3
in Table~2 and~\eqref{reprandacc}), with total savings
$\Omega(ds\log\frac{n}{s})$. This first step corresponds to bounding the
partial cost $ \bin(S) + \bin(A_S^{\mbox{out}}) + \bin(\dest(A_S))$
and it is provided in Lemma \ref{cl:sav_accepted}.

If this is intuitively the key argument, it is neglecting the fact that we now
need to identify the subset $A_S$ of requests originating from $S$ that are accepted
(see Field~1 of Table~2 and~\eqref{repaccs}). For each node, this cost depends
on the number of failures and in general cannot be compensated by the
aforementioned savings. To this purpose, we need to exploit the further
property that, for each node, failures have destinations that, for each round
$t$, correspond to an $\bigO(1/c)$ fraction of the vertices. Thanks to the use
of semisatured and critical nodes (see Field~4 of Table~2), compressing the
destinations of these requests allows to compensate the aforementioned cost
almost entirely. This second step corresponds to bounding the partial cost
$\bin(A_S) +  \bin(C) + \bin(\dest(\sRJ))$ and it is provided in Lemma \ref{lemma:ub_second}.

We state and prove a useful bound, that easily follows from the expansion
property of the underlying graph $G$ and Lemma~\ref{le:exp_mixing_lemma}.

\begin{lemma}\label{lemma:s_big}
Let $G = (V, E)$ be a $\Delta$-regular graph and let $\lambda$ be the second
largest eigenvalue of $G$'s adjacency matrix. Then, for any subset $S \subseteq
V$, it holds that
\[
1 - \out \, \le \, \frac{s}{n} + \frac{\lambda}{\Delta} \, .
\]
where $s = |S|$ and $\delta$ is defined as in~\eqref{eq:defeps&delta}.
\end{lemma}

\begin{proof}
From the definition of $\out_v$ we have that the number of edges with both
end-points in $S$ is 
\[
e(S, S) = \sum_{v\in S}(1 - \out_v)\Delta = (1 - \out)s\Delta \, .
\]
From Lemma~\ref{le:exp_mixing_lemma} it thus follows that
\[
(1 - \out) = \frac{e(S,S)}{s\Delta } \leq \frac{1}{2 s\Delta}
\left(\frac{\Delta s^2}{n}  + \lambda s \right) = \frac 12 \left(\frac{s}{n} +
\frac{\lambda}{\Delta} \right) \, . 
\]
\end{proof}

As a first, crucial  step of our compression analysis,  we evaluate the cost
$\bin(\dest(A_S))$ - see ~\eqref{reprandacc} - of representing the destinations
corresponding to the subset $A_S$ of accepted link requests from nodes in $S$.
We decided to isolate this step since it is the only one in which we make use
of the expansion property of the underlying graph $G$, stated in the   lemma
above.

\begin{lemma}[Bounding $\bin(\dest(A_S))$]\label{cl:sav_accepted}
Under the hypotheses of Theorem~\ref{thm:main-expanders}, the   cost $\bin(\dest(A_S))$ - see ~\eqref{reprandacc} - of
representing the destinations corresponding to the subset $A_S$ of accepted
link requests from nodes in $S$ satisfies the following bound:
\begin{equation}\label{eq:sav_accepted}
 \bin(\dest(A_S)) \leq \sum_{v\in S}(1 - \freq_v)\nreq\log ((1 - \out_v)\Delta) 
+  \freq_v\nreq\log\Delta\le sd\log\Delta 
- \frac{1 - \freq}{2}sd\log\frac{n}{s} + 2\freq ds \,.
\end{equation}
\end{lemma}
\begin{proof}
We directly give a lower bound to the savings achieved with respect to the cost of  
the uncompressed representation, the latter being  $sd\log\Delta$.  Namely, we prove that
\begin{equation}\label{eq:sav_accepted-1}
sd\log\Delta - \left(\sum_{v\in S}(1 - \freq_v)\nreq\log ((1 - \out_v)\Delta) 
+ \sum_{v\in S}\freq_v\nreq\log\Delta\right)
\ge \frac{1 - \freq}{2}ds\log\frac{n}{s} - 2\freq ds \, ,
\end{equation}
whence \eqref{eq:sav_accepted} immediately follows. 
First of all,  the LHS of \eqref{eq:sav_accepted-1} can be written as 

\begin{align}
& \sum_{v \in S} d\log\Delta - \sum_{v\in S}\freq_v \nreq \log\Delta  -
\left(\sum_{v\in S}(1 - \freq_v)\nreq\log ((1 - \out_v)\Delta) \right) 
  \nonumber \\
& = \nreq\sum_{v\in S}(1 - \freq_v)\log\Delta - \nreq\sum_{v\in S}(1 - \freq_v)\log((1 - \out_v)\Delta) 
= \nreq\sum_{v\in S}(1 - \freq_v)\log\frac{1}{1 - \out_v} \, . \label{eq:gain1}
\end{align}

Next, we consider two cases.

\smallskip\noindent
\underline{Case $s < \alpha\Delta$.} By definition of  
 $\out_v$, we easily get    that 
\[(1-\out_v) \Delta =  e_G(v,S) \leq s \]
that immediately implies  
\begin{equation} \label{obs:s_small}
1 - \out_v\le\frac{s}{\Delta} \, .
\end{equation}
From \eqref{eq:gain1}, we get 
\begin{align*}
& \nreq\sum_{v\in S}(1 - \freq_v)\log\frac{1}{1 - \out_v}
\ge\nreq\sum_{v\in S}(1 - \freq_v)\log\frac{\Delta}{s} 
= (1 - \freq)sd\log\frac{\Delta}{s} 
> \frac{1 - \freq}{2}sd\log\frac{n}{s} \, ,
\end{align*}
where we used \eqref{obs:s_small} to write the first 
inequality, while the last inequality follows from the  definition  
$\freq = (1/s)\sum_{v \in S} \freq_v $ and since, in this case, 
\[
\frac{\Delta}{s} > \frac{1}{\alpha} = \frac{n}{\Delta} \, ,
\]
which in turn implies
\[
\frac{\Delta}{s} \cdot \frac{\Delta}{s} = \left(\frac{\Delta}{s}\right)^2 > 
\frac{n}{\Delta} \cdot \frac{\Delta}{s} = \frac{n}{s} \, .
\]

\smallskip\noindent
\underline{Case $\alpha\Delta\le s\le n/2$.} In this case, the proof is a
bit more articulated. To begin, we can write
\begin{align}
& \sum_{v\in S}(1 - \freq_v)\log\frac{1}{1 - \out_v} 
= -\sum_{v\in S}(1 - \freq_v)\log(1 - \out_v) 
= -(1 - \freq)s\sum_{v\in S}\frac{1 - \freq_v}{(1 - \freq)s}\log(1 - \out_v) \nonumber\\
& \ge -(1 - \freq)s\log\frac{\sum_{v\in S}(1 - \freq_v)(1 - \out_v)}{(1 
	- \freq)s}
\ge-(1 - \freq)s\log\frac{\sum_{v\in S}(1 - \out_v)}{(1 - \freq)s} 
= (1 - \freq)s\log\frac{1 - \freq}{1 - \out} \, . \label{eq:jensappli}
\end{align}
Here, to derive the third inequality we used Jensen's inequality on the function 
\[ \sum_{v\in
S}\frac{1 - \freq_v}{(1 - \freq)s}\log(1 - \out_v) \, , \] which is a convex
combination of values of a concave function. 

Next, going back to \eqref{eq:gain1}, from \eqref{eq:jensappli} and Lemma~\ref{lemma:s_big} we easily get  
\begin{equation} \label{eq:II-case-a}
	\nreq\sum_{v\in S}(1 - \freq_v)\log\frac{1}{1 - \out_v}  \ge
	 (1 - \freq)ds\log\frac{1 - \freq}{1 - \out} \ge
	(1 - \freq)ds\log\frac{(1 -  \freq)n\Delta}{\lambda n + \Delta s}
\end{equation}

Since we are assuming  $s\ge\alpha\Delta$,  for any $\lambda\le\freq\alpha^2\Delta$, it holds that $\lambda\le\freq\alpha^2\Delta$, and, hence,
\[
\frac{1}{s}\le (1 + \freq)\frac{\Delta}{s\Delta + \lambda n} \, .
\]
This, in turn, implies that 
\[
\frac{(1 - \freq)n\Delta}{\lambda n + \Delta s}\ge 
\frac{n}{s}\cdot\frac{1 - \freq}{1 + \freq}\, .
\]
From the above inequality and \eqref{eq:II-case-a}, we easily get

\begin{align*}
& \nreq\sum_{v\in S}(1 - \freq_v)\log\frac{1}{1 - \out_v}   \ge   (1 - \freq) ds \log \left( \frac{n}{s}\cdot\frac{1 - \freq}{1 + \freq}\right)  \\
& \ge          (1 - \freq)ds\log \frac{n}{s} - (1 - \freq)ds \log  \left(1 + \frac{2 \freq}{1 + \freq} \right) 
   \ge    (1 - \freq)ds\log \frac{n}{s} - (1 - \freq) ds \left(\frac{2 \freq}{1 + \freq} \right)   \\
& \ge (1 - \freq)ds\log \frac{n}{s} -  ds  2 \freq \, .
\end{align*}
\end{proof}

The above lemma is then used below to bound the first partial cost of our
compressed representation of $R$.

\begin{lemma}[Bounding $ \bin(S) + \bin(A_S^{\mbox{out}}) + \bin(\dest(A_S))$]
\label{lemma:ub_first}
Under the hypotheses of Theorem~\ref{thm:main-expanders}, the overall cost of
representing sets $S$, $A_S^{\mbox{out}}$ and $\dest(A_S)$ - see \eqref{reps},
\eqref{repcut}, and \eqref{reprandacc} - satisfies the following bound

\begin{equation} \label{eq:bound-overall-I}
\bin(S)   + \bin(A_S^{\mbox{out}} )  + \bin(\dest(A_S)) \leqslant
3s\log\frac{n}{s} + ds\log\Delta - \frac{1 - 13\freq\log(1/\freq)}{2} ds\log\frac{n}{s} + 2\freq ds
\end{equation}
\end{lemma}
\begin{proof}
As for $\bin(S)$, from~\eqref{reps}, notice that  
\begin{equation} \label{eq:boundonS}
2\log s + \log {n \choose s}\le 3s\log\frac{n}{s} \, ,
\end{equation}
since we are assuming $s \le n/2$ and it holds $\log {n \choose s} \leq s \log
(ne/s)$.

As for $\bin(A_S^{\mbox{out}})$, the first term in~\eqref{repcut} can be
bounded as follows 
\begin{equation} \label{eq:boundonAsv-I}
	2\sum_{v\in S}\log(\freq_v d) = 2\log\prod_{v\in S}(\freq_v 
	d)
	\le 2\log\left(\frac{\sum_{v\in S}\freq_v d}{s}\right)^s = 2s\log(\freq d)\, ,
\end{equation}
where the second inequality follows from the AM-GM
inequality~\cite{steele2004cauchy}, while the last equality follows from the
definition of $\freq$ (seei~\eqref{eq:fractions}). Moreover, the second term
in~\eqref{repcut} can be bounded as follows 
\begin{equation}\label{eq:boundonAsv-II}
	\sum_{v\in S}\log\binom{d}{\freq_v 
	d}\le d\sum_{v\in S}\freq_v\log\frac{e}{\freq_v} = \freq ds\log e + 
	d\sum_{v\in S}\freq_v\log\frac{1}{\freq_v}\le\freq ds\log e + \freq 
	d s\log\frac{1}{\freq} \, .
\end{equation}
Here, the second equality follows again from the 
definition of $\freq$, while the third inequality follows since the optimum 
of the following problem
\begin{align*}
	&\max g(x_1, ... , x_k) = \sum_{i=1}^k x_i\log\frac{1}{x_i}\\
	&\sum_{i=1}^k x_i = B 
	 \ \mbox{ with } x_i\ge 0\, , \ i = 1, \ldots , k
\end{align*}
is achieved when $x_1 = \ldots = x_k = B/k$.

Finally, combining~\eqref{eq:boundonS}, \eqref{eq:boundonAsv-I},
\eqref{eq:boundonAsv-II}, and~\eqref{eq:sav_accepted} given in Lemma
\ref{cl:sav_accepted}, we get
\begin{align}\label{eq:partial1}
	& \bin(S)   + \bin(A_S^{\mbox{out}} )  + \bin(\dest(A_S)) \nonumber\\ 
	& \le 3s\log\frac{n}{s} + \underbrace{2s\log(\freq d) + \freq ds\log e 
	 + \freq  d s\log\frac{1}{\freq}}_{*} 
	 + ds\log\Delta - \frac{1 - \freq}{2}ds\log\frac{n}{s} + 2\freq ds \nonumber \\
	& \le 3s\log\frac{n}{s} + ds\log\Delta - \frac{1 - 13\freq\log(1/\freq)}{2} ds\log\frac{n}{s} + 2\freq ds\, ,
\end{align}
where the last inequality holds, since each of the starred terms is at 
most $2 \freq \log(1/\freq) ds\log\frac{n}{s}$, whenever $s  \le n/2$,  and $\freq \leq 1/2$.

\end{proof}

In the next lemma, we bound the remaining part of our compressed representation of $R$. 

\begin{lemma}[Bounding  $\bin(A_S) +  \bin(C) + \bin(\dest(\sRJ))$]
\label{lemma:ub_second}
Under the hypotheses of Theorem~\ref{thm:main-expanders}, the overall cost of
representing sets  $A_S$,  $\bin(C)$, and $\dest(\sRJ))$ - see \eqref{repaccs},  \eqref{repcrit}, and
\eqref{reprandrej} - satisfies the following bound
\[
\bin(A_S) +  \bin(C) + \bin(\dest(\sRJ))
\leqslant
\log\Delta\sum_{v\in S}(\ell_v - d) + \frac{1}{4}ds \, .
\]
\end{lemma}

\begin{proof}
As for $\bin(A_S)$, from $\eqref{repaccs}$, we next show that  
\begin{equation} \label{eq:costAs}
     \bin(A_S) = \sum_{v\in S}\left(2\log\ell_v + 
\log\binom{\ell_v}{d}\right)
\le \sum_{v \in S} \left( 5d(\ell_v - d) + \frac{d}{4} \right)\, .
\end{equation}
%
%
When $\ell_v=d$, the last inequality follows from the hypothesis $d\geq 44$, 
since $2\log\ell_v + 
\log\binom{\ell_v}{d} = 2\log d \leq \frac d4$. 
When $\ell_v>d$, we note that 
\[
 5d \geq \underbrace{ d(\log e +1)}_{(\star)}  + \underbrace{2d}_{(\star \star)} ,
\] and we observe that $(\star \star)$ accounts for the first term of $\bin(A_S)$
\[
2d(\ell_v - d) \geq 2\log \ell_v,
\]
and that $(\star)$ accounts for the second term of $\bin(A_S)$
\[
\log \binom{\ell_v}{d} 
\leq d\log \frac{e\ell_v}{d}
= d\log e + d(\log \ell_v - \log d) 
\leq d\log e + d(\ell_v - d) = d(\log e +1)(\ell_v - d).
\]

\smallskip
As for $ \bin(C)$, 
the definition of critical node at round $t$ implies
that each critical node at round $t$ is receiving more than $cd/2$ of its
incoming requests from nodes in $S$. As a consequence, we get
\[
\sum_{v\in S}\rc_t(v) > \frac{cd}{2}c_t \, 
\]
where we recall that $\rc_t(v)$ is the number of rejected requests   from node $v$ to critical nodes and $c_t$ is the size of the subset of nodes which turn to be critical at round $t$.
For the first term in \eqref{repcrit} this implies
\begin{equation}\label{eq:costC-I}
	2\sum_{t=1}^T\log c_t\le 
	2\sum_{t=1}^T\log\left(\frac{2}{cd}\sum_{v\in S}\rc_t(v)\right) < 
	\frac{4}{cd}\sum_{t=1}^T\sum_{v\in 
	S}\rc_t(v)\le\frac{4}{cd}\sum_{v\in S}(\ell_v - d) \, ,
\end{equation}
where to derive the last inequality we exchanged the order of summation 
and used the fact that $\sum_{t=1}^T\rc_t(v)\le \ell_v - d$.   
As for     the second term in \eqref{repcrit}, we get
\begin{equation}\label{eq:costC-II}
	\log\binom{n}{c_t}\le c_t\log\frac{en}{c_t}\le\frac{2}{cd}\sum_{v\in 
	S}\rc_t(v)\log\frac{en}{c_t}\le\sum_{v\in 
	S}\rc_t(v)\log\frac{2n}{c\cdot c_t}.
\end{equation}
Here, to derive the third inequality we used the following
\vspace{-24pt}
\begin{quote}
\begin{claim}\label{cl:tech1}
It holds
\[
\log\left(\frac{en}{c_t}\right)^{\frac{2}{cd}}\le\log\frac{2n}{c\cdot c_t}\, ,
\]
whenever $c$ is large enough that $d\ge\frac{2}{c}\log\frac{ce}{2}$.
\end{claim}
\begin{proof}
The proof follows from simple calculus, recalling that the definition 
of critical node implies that there are at most $n/c$ of them at any 
round $t$.
\end{proof}
\end{quote}

\noindent
By combining \eqref{eq:costAs}, \eqref{eq:costC-I},   \eqref{eq:costC-II} and \eqref{reprandrej},  
 we set $\gamma = 1 + 5d + \frac{4}{cd}$ and  
 derive  the following   upper bound
\begin{align}\label{eq:partial2}
	&\bin(A_S) +  \bin(C) + \bin(\dest(\sRJ)) \nonumber \\
	& \le  \sum_{v\in S}\left(\gamma(\ell_v - d) + 
	\frac{1}{4}d + \sum_{t=1}^T
	\rc_t(v)\log\frac{2n}{c\cdot c_t} + \rss(v) \cdot \log 
	\frac{2n}{c} + \sum_{t=1}^T \rc_t(v) \cdot \log c_t \right)
	\nonumber\\
	& {=} \gamma\sum_{v\in S}(\ell_v - d) + \frac{1}{4}ds + \sum_{v\in S}\left( 
	     \sum_{t=1}^T\rc_t(v)\log\frac{2n}{c} +
	     \rss(v) \cdot \log \frac{2n}{c} \right)\nonumber\\
	& {=} \gamma\sum_{v\in S}(\ell_v - d) + \frac{1}{4}ds 
	    +  \sum_{v\in S}(\ell_v - d)\log \frac{2n}{c} \nonumber\\
	& {=} \gamma\sum_{v\in S}(\ell_v - d) + \frac{1}{4}ds 
	    +  \sum_{v\in S}(\ell_v - d)\log \frac{2n}{\sqrt c} - \frac{1}{2}\sum_{v\in S}(\ell_v - d)\log c \nonumber\\
	& \stackrel{(a)}{\leq} \gamma\sum_{v\in S}(\ell_v - d) + \frac{1}{4}ds + 
        \sum_{v\in S}(\ell_v - d) \log\Delta - \frac{1}{2}\sum_{v\in 
	S}(\ell_v - d)\log c\nonumber\\
	&\stackrel{(b)}{\leq}\log\Delta\sum_{v\in S}(\ell_v - d) + \frac{1}{4}ds \, ,
\end{align}
where we used that we can assume $c$ large enough so that, 
in $(a)$, $c\ge (2/\alpha)^2$ and $\Delta = \alpha n$ imply $\Delta\ge 2n/\sqrt{c}$, while in $(b)$
we used $\gamma = 1+5d +\frac{4}{cd}\leq \frac 12\log c$.
\end{proof}

\paragraph{Wrap up: Proof of Theorem~\ref{thm:main-expanders}.}
From Lemmas~\ref{lemma:ub_first} and~\ref{lemma:ub_second} it follows that the
total number of bits to encode the executions of nodes in $S$ is as follows
(recall that we use $(n - s)dT\log\Delta$ bits for vertices in $V - S$).
\begin{align}\label{eq:total}
& \bin(S) + \bin(A_S^{\mbox{out}}) + \bin(\dest(A_S)) + 
\bin(A_S) +  \bin(C) + \bin(\dest(\sRJ))
 \nonumber\\ 
& \le 3s\log\frac{n}{s} + ds\log\Delta - \frac{1 - 13\freq\log(1/\freq)}{2} ds\log\frac{n}{s} + 2\freq ds
+ \log\Delta\sum_{v\in S}(\ell_v - d) + \frac{1}{4}ds \, 
\end{align}
To this cost, we should also add the randomness that was not used, 
which is exactly $(dT - \ell_v)\log\Delta$ for the generic node $v$. 
Our savings are then 
\begin{align*}
	&\mbox{Savings}\ge dsT\log\Delta \nonumber\\
	&\hspace{12pt}- \left(3s\log\frac{n}{s} + ds\log\Delta - \frac{1 - 13\freq\log(1/\freq)}{2} ds\log\frac{n}{s} + 2\freq ds
	 + \log\Delta\sum_{v\in S}(dT - d) + \frac{1}{4}ds\right)\\
	& = - 3s\log\frac{n}{s} + \frac{1 - 13\freq\log(1/\freq)}{2} ds\log\frac{n}{s} - 
	\left(\frac{1}{4} + 2\freq\right)ds \, ,
\end{align*}
Finally, the expression above is $\Omega(\log n)$, as soon as $\freq$ is a
sufficiently small (absolute) constant.

\section{Future work} \label{sec:concl}
A first interesting open problem is extending the analysis of \ALG\ to
non-dense expanders, i.e., to cases in which $\Delta = o(n)$. In this setting,
both the proofs of convergence and of expansion might need to be revisited in
significant ways. For example, if $\Delta < n/c$, we can no longer guarantee
that all nodes will eventually establish $d$ connections: it might well be the
case that all neighbours of some node become saturated at some point, before
the node itself can see all its requests accommodated. In fact, $\Delta =
\Omega(\log n)$ is necessary to ensure that this does not occur with high
probability. Another interesting generalization is extending the analysis to
the case of non-regular graphs, possibly relying on the corresponding
generalization of the Expander Mixing Lemma. Finally, it would be interesting
to investigate the robustness of \ALG\ in dynamic settings, in which nodes
and/or vertices of the underlying graph $G$ may join or leave the network.

\bibliographystyle{plain}
\bibliography{drrs}

\newpage
\appendix
\begin{center}
\LARGE{\textbf{Appendix}}
\end{center}
\section{Mathematical Tools} \label{app:tools}
In Section~\ref{subsec:expconvtime}, we used the method of bounded differences.
In particular, we applied - a slight generalization of\footnote{Via a simple coupling argument,
in the bound in \cite{DP09,Mc98} we can substitute to the expected value an upper bound to it.} - the concentration bound in \cite{DP09,Mc98}. 

\begin{theorem} \label{thm:mobd}
Let $\mathbf{Y} = (Y_1, \ldots, Y_m)$ be independent r.v.s, with  $Y_j$ taking
values in a set $A_j$. Suppose the real-valued funcion $f$ defined on $\prod
A_j$ satisfies 
\[
|f(\by) - f(\by')|  \leq \beta_j
\]
whenever vectors $\by$ and $\by'$ differs only in the $j$-th coordinate. Let
$\mu$ be an upper bound to the expected value of r.v. $f(\mathbf{Y})$.  Then, for any $M>0$, it
holds that
\[
\Pr\left( f( \mathbf{Y}) - \mu \, \geq \,  M  \right) \, \leq \, e^{-
\frac{2M^2}{\sum_{j=1}^m \beta_j^2}} \, .
\]
\end{theorem}

\section{Proving Lemma~\ref{lem:termination} via an encoding argument}
\label{app:time}
We provide here an elegant, alternative proof for the fact \ALG\ on graph  $G$
completes its task within a logarithmic number of rounds, w.h.p. The proof
relies on a simple encoding argument~\cite{MMR17} and it can be seen as a
``warm-up'' for  the much more complex analysis given  in
Section~\ref{sub:expexp} which makes use of the same approach.

\begin{lemma}\label{le:exp_completion}
Let $G = (V,E)$ be any $\Delta$-regular graph with $\Delta = \alpha n$ and $0 <
\alpha \leqslant 1$ and let $d\geq 1$ be any absolute constant. Then, for any
$c > 1/\alpha$, any $\beta > 2$, and any large-enough $n$, \ALG$(G,d,c)$ on
graph $G$ terminates within $(\beta / \log (\alpha c)) \log n$ rounds with
probability at least $1-n^{-(\beta - 2)/2}$.
\end{lemma}

\begin{proof}
We prove the lemma via an \textit{encoding} argument~\cite{MMR17}\footnote{As remarked in Section \ref{ssec:description}, we can avoid taking ceilings of the quantities measuring the number of bits for the encoding.}. Notice that
any execution of \ALG\ for $t$ rounds is completely determined by a sequence of
$t n d \log \Delta$ random bits: Indeed, $\log \Delta$ random bits can be used
for each link request of any fixed node, each node makes at most $d$ link
requests at each round, and this procedure is repeated for every node and for
$t$ rounds.

Consider an execution where there is a node $v$ with $d' < d$ outgoing edges at
round $t$ and note that this node must have had at least one rejected request
in each  of the $t$ rounds. We can encode the sequence of $t n d \log \Delta$
bits that generates such an execution as follows: The first $\log n$ bits
encode node $v$; The following $t (n-1) d \log \Delta$ bits encode all possible
random choices of all nodes but $v$; As for the random bits of node $v$, let
$\ell_v$ be the number of random choices actually made by $v$ during the $t$
rounds; we can use

\begin{itemize}
\item $2 \log \ell_v$ bits to encode in a prefix-free way the number $\ell_v$;
\item $2 \log d'$ bits to encode in a prefix-free way the number $d'$ of 
accepted requests;
\item $\log \binom{\ell_v}{d'}$ bits to encode the positions of the accepted 
requests in the sequence of $\ell_v$ requests;

\item $\log \Delta$ bits for each one of the $d'$ accepted request and
$\log(n/c)$ bits for each one of the $\ell_v - d'$ rejected ones. Notice that
we can use only $\log (n/c)$ bits instead of $\log \Delta$ to encode a rejected
request since (i) each rejected request was directed to a node that was
\textit{overloaded} (i.e., a node that received at least further $cd$ incoming
requests) at the time of the request from $v$; (ii) since each node sends at
most $d$ requests at each round, it follows that the number of overloaded nodes
is at most $n/c$ at each round; (iii) the previous bits of our encoding
uniquely identify the set of overloaded nodes at each round. 
\end{itemize}

\noindent In contrast to the $\ell_v \log \Delta$ bits used by node $v$ in the
uncompresed encoding we thus use only
\begin{align*}
& 2 \log \ell_v + 2 \log d' + \log \binom{\ell_v}{d'} + d' \log n + (\ell_v -
d')\log(n/c) = \\
& = \ell_v \log (n/c) + 2 \log (\ell_v d') + d' \log c + \log
\binom{\ell_v}{d'} = \\
& = \ell_v \log \Delta - \left[ \ell_v \log (\alpha c) - 2 \log (\ell_v d') -
d' \log c - \log \binom{\ell_v}{d'} 
\right]
\end{align*}

\noindent Hence, the total number of bits for node $v$ saved in our encoding is
\begin{align*}
& \ell_v \log (\alpha c) - d' \log c - \log \binom{\ell_v}{d'} - 2 \log (\ell_v
d') \\
& \geq \ell_v \log(\alpha c) - d' \log c - d' \log \frac{e \ell_v}{d'} - 2 \log
(\ell_v d') \\
& \geq (1/2) \ell_v \log (\alpha c) 
\end{align*}
where in the first inequality we used that $\log \binom{\ell_v}{d'} \leqslant
d' \log(e \ell_v / d')$ and the last inequality holds for $\ell_v$ large
enough, since $c$, $d'$, and $\alpha$ are $\bigO(1)$. 

By using that $\ell_v \geqslant t \geqslant (\beta/\log(\alpha c)) \log n$ and
that in our encoding we use $\log n$ bits to identify node $v$, the fraction of
strings determining an execution such that there is a node $v$ with $d' < d$
outgoing edges at round $t$ is thus at most 
\[
2^{- (1/2) \ell_v \log (\alpha c) + \log n} \leqslant 2^{-(1/2) t \log (\alpha
c) + \log n} \leqslant2^{-(1/2) \beta \log n + \log n} = n^{-(\beta - 2)/2} \,. 
\]
\end{proof}

\end{document}